\documentclass[12pt,draftcls,onecolumn]{IEEEtran}
\usepackage{amsmath,amssymb,array,graphicx,verbatim}

\usepackage{latexsym,url}
\usepackage{amsmath,amssymb,amsthm}
\usepackage{algorithm,algorithmic}
\usepackage{color,array,graphicx,verbatim,xtab}
\usepackage{cite} 
\usepackage{setspace}

\usepackage
  [breaklinks,bookmarks,bookmarksnumbered,bookmarksopen,bookmarksopenlevel=2]
  {hyperref}

\newtheorem{theorem}{Theorem}
\newtheorem{lemma}[theorem]{Lemma}
\newtheorem{definition}{Definition}
\newtheorem{example}{Example}

\usepackage{color}
\newcommand{\mycomment}[3]%
{%
\marginpar{%
  \hfil%
  \tiny{\textcolor{#2}{{\bf\textsc{#1}}}}%
  \hfil%
}%
\footnote{\textcolor{#2}{{\bf\textsc{#1}:}~~#3}}
}

\begin{document}

\title {On the Construction of Prefix-Free and Fix-Free Codes with Specified Codeword Compositions}

\author
{Ali Kakhbod\thanks{Ali Kakhbod is with the department of
  Electrical Engineering and Computer Science (EECS), University of Michigan, Ann Arbor, MI, USA. (email:
 \texttt{akakhbod@umich.edu}).}
 and Morteza Zadimoghaddam \thanks{Morteza Zadimoghaddam is with the Computer Science and Artificial Intelligence Laboratory (CSIL), MIT, Cambridge, MA, USA. (email:
 \texttt{morteza@mit.edu}). }
 }


\maketitle
\begin{abstract}
We investigate  the construction of prefix-free and fix-free codes with specified codeword compositions.  We present  a polynomial time algorithm which constructs a fix-free code with the same codeword compositions as a given code for a special class of codes called distinct codes.  We consider the construction of optimal fix-free codes which minimize the average codeword cost for general letter costs with uniform distribution of the codewords and  present an approximation algorithm to find a near optimal fix-free code with a given constant cost.
\\

\end{abstract}
\begin{section}{Introduction}
The basic elements of  a discrete communication system are its source, encoder, channel, decoder and destination. The source may be represented as a random variable, $X$, taking on values from the set of source characters $\{x_1,x_2,\cdots,x_M\}$ with probabilities $p_1,p_2,\cdots,p_{M}$, respectively. A message is a sequence of source characters.  To facilitate transmission, the  encoder associates with every source character, $x_i$, a finite sequence of  code characters $a_1,a_2,\cdots,a_D$ (D-ary). Such a sequence of code characters is called a codeword. A code, denoted by $S$, is the collection of all codewords. The encoded message is then transmitted over the channel which we assume to be noiseless. At the receiving end, the decoder attempts to reproduce the original message by assigning a set of source characters to the coded message. \\   

To avoid ambiguity, every finite sequence of code characters  must correspond to no more than one message. A code that conforms with  this requirement  is said to be a \textit{uniquely decodable} code. Furthermore, to simplify the decoding procedure,  two other type of codes are often used in communication systems defined as follows.   If  no codeword is a prefix to some other codeword, the code is said to be a \textit{prefix-free} code, and  if  no codeword is a prefix  or suffix to some other codeword, the code is said to be a  \textit{fix-free} code. We denote the set of all codes, uniquely decodable codes, prefix-free codes and fix-free codes,  that can be constructed from the code character $\{a_1,a_2\cdots,a_D\}$, by $\mathcal{C}^D, \mathcal{C}_{ud}^D,$ $\mathcal{C}_{pf}^D$ and $\mathcal{C}_{ff}^D$, respectively. Along the paper, superscript $D$ is omitted for binary codes. In general, directly from definitions, it can be deduced that $\mathcal{C}^D \supset \mathcal{C}_{ud}^D \supset \mathcal{C}_{pf}^D \supset \mathcal{C}_{ff}^D.$ We illustrate it with the following example. 
\\

\begin{example}
Consider the following four binary codes,
$$
S_1=\{00,10,11\}
$$
$$
S_2=\{00,10,11,011\}
$$
$$
S_3=\{00,10,11,110,100\}
$$
$$
S_4=\{0,001,100,110\}.
$$
 $S_1$ is a fix-free code ($S_1 \in \mathcal{C}_{ff}$), $S_2$ is a prefix-free code but is not fix-free ($S_2 \in \mathcal{C}_{pf}, S_2 \notin S_{ff}$), $\mathcal{C}_3$ is a uniquely decodable code but is neither prefix-free nor fix-free ($S_3 \in \mathcal{C}_{ud}, S_3 \notin \mathcal{C}_{pf}, S_3 \notin \mathcal{C}_{ff}$) but $\mathcal{C}_4$ is neither uniquely decodable,   prefix-free nor fix-free ($S_4 \in S,  S_4 \notin \mathcal{C}_{ud}, S_4 \notin \mathcal{C}_{pf}, S_4 \notin \mathcal{C}_{ff}$). 
\end{example}

Let $S=\{s_1,s_2,\cdots,s_n\}$ be a code. The composition of a codeword $s_k, k=1,2,\cdots,n,$ is written as $\big(\delta_1^{(k)},\delta_2^{(k)},\cdots,\delta_D^{(k)} \big)$ where $\delta_i^{(k)}$ is the number of times the code character $a_i$ appears in the codeword $s_k$. Suppose that a set of costs $\{c_1,c_2,\cdots,c_D\}$  associated with the respective code characters $\{a_1,a_2,\cdots,a_D\}$, i.e. $c_i$ is positive  corresponding to $a_i, i=1,2,\cdots,D$, then the average codeword cost of the code $S$ is equal to 

\begin{eqnarray}
\label{eqn1}
\sum_{k=1}^{n}p_k \left[\sum_{j=1}^D {\delta_j^{(k)}c_j}\right]
\end{eqnarray}
 where $p_k$ is the probability assigned to $s_k, k=1,2,\cdots,n$. \\

\begin{example}
The message $\alpha\beta\rho\alpha\kappa\alpha\delta\alpha\beta\rho\alpha\alpha\alpha\rho$ can be considered to be a 6-ary
message over the alphabet $\{\alpha,\beta,\kappa,\delta,\chi,\rho\}$. Its length is 14, and
its composition vector is $(7,2,1,1,0,3)$. Assuming respective
symbol costs $(1,3,3,2,10,1)$ then the cost 
is 21.
\end{example}

It is known that for equal costs, i.e.,  $c_1=c_2=\cdots=c_D$,  Huffman's algorithm, \cite{huff}, derives an optimal prefix-free code,  but when the costs $c_1,c_2,\cdots,c_D$ are not all equal, the composition of the codewords becomes important.
The problem of constructing optimal code for minimizing the average cost has been considered for prefix-free codes in  \cite{Alter,golin,golin2008}.
Constructing optimal fix-free codes with the aim of minimizing the average code length, equal letter costs, is recently considered in \cite{serap2}.
Upper bounds on the average code length of optimal fix-free codes which minimize the average code length for equal letter cost, but general probability distributions of the alphabet symbols are provided in \cite{ahl,ye} (in contrast, in this work, we consider the construction of optimal fix-free codes which minimize the average codeword cost for general letter costs with uniform distribution of the codewords). \\

As mentioned  in  above, when costs are unequal then the  composition of the codewords plays an important role in constructing optimal codes. In this paper, we provide a necessary and sufficient condition for the existence of a D-ary prefix-free code with a given set of compositions (this is an immediate extension of Proposition 2 of \cite{carter} to D-ary codes) and then we present a polynomial algorithm that results in a binary prefix-free code with the same composition set of a given code.  We also  present an algorithm to find a fix-free code for a given set of compositions of  a special class of codes  that we call  \textit{distinct} codes, if such a fix-free code exists. Consequently, we present an approximation algorithm to find a near optimal fix-free code with a given constant cost.  All the results refer to binary codes.

\end{section}

\begin{section}{Prefix-free codes}
\label{sec1}
In the following, we present a necessary and sufficient condition for existence of a D-ary prefix-free code with a given set of codeword compositions which is an immediate extension  of Proposition 2 of \cite{carter} to D-ary codes. Then,  we establish a polynomial time algorithm to find a binary prefix-free code with a given composition set.


\begin{theorem}[\cite{carter}]
\label{1}
Let $\Delta=\{\big(\delta_1^{(k)},\delta_2^{(k)},\cdots,\delta_D^{(k)} \big), \ 1 \leq k \leq n\}$  be the set of codeword  compositions of some code $S$ (with $n$ codewords). Then there exists a prefix-free code with the same set of codeword compositions if and only if the following inequality holds for each $(\delta_1^{(k)},\delta_2^{(k)},\cdots,\delta_D^{(k)}) \in \Delta$, (length of any codeword $s_k \in S, 1 \leq k \leq n,$ is denoted by $l_k$, i.e. $l_k:=\sum_{i=1}^D\delta_i^{(k)}$)
\begin{align}
\prod_{i=1}^{D-1} \binom{\sum_{j=i}^D \delta_j^{(k)}}{\delta_i^{(k)}} \geq   
   \sum_{t=1}^{l_k}{\sum_{\substack{\xi_1^{(k)}+\xi_2^{(k)}+\cdots + \xi_D^{(k)}=t \\
\xi_i^{(k)} \leq \delta_i^{(k)}}}  \Lambda_{\xi_1^{(k)},\xi_2^{(k)},\cdots,\xi_D^{(k)}} \prod_{r=1}^{D-1}{\binom{\sum_{i=r}^D  \left(\delta_i^{(k)}-\xi_i^{(k)}\right)}{\delta_r^{(k)}-\xi_r^{(k)}}}} 
\end{align}
where $ \Lambda_{\xi_1^{(k)},\xi_2^{(k)},\cdots,\xi_D^{(k)}} $ is the number of codewords of
composition  $(\xi_1^{(k)},\xi_2^{(k)},\cdots,\xi_D^{(k)})$ in $S$.
\end{theorem}

\begin{proof}
The number of all codewords of composition  $(\delta_1^{(k)},\delta_2^{(k)},\cdots,\delta_D^{(k)})$ is $\prod_{i=1}^{D-1} \binom{\sum_{j=i}^D \delta_j^{(k)}}{\delta_i^{(k)}}$. In addition, it is clear that,  the number of words of composition $\big(\delta_1^{(k)},\delta_2^{(k)},\cdots,\delta_D^{(k)}\big)$ with a prefix code of composition  $(\xi_1^{(k)},\xi_2^{(k)},\cdots,\xi_D^{(k)})$ is  $\prod_{r=1}^{D-1}{\binom{\sum_{i=r}^D  \left(\delta_i^{(k)}-\xi_i^{(k)}\right)}{\delta_r^{(k)}-\xi_r^{(k)}}}$.  Therefore,  the necessity of the theorem is resulted when the number of all codewords of composition $(\delta_1^{(k)},\delta_2^{(k)},\cdots,\delta_D^{(k)})$ is greater than the number of codewords of composition $(\xi_1^{(k)},\xi_2^{(k)},\cdots , \xi_D^{(k)})$ which must be removed by the prefix condition.\\ 

To prove the sufficiency of the theorem, we construct a prefix code with the given composition by an algorithm.  We start from shorter codewords, at each iteration if we need 
$\Lambda_{\delta_1^{(k)},\delta_2^{(k)},\cdots,\delta_D^{(k)}}$ codewords of composition $(\delta_1^{(k)},\delta_2^{(k)},\cdots,\delta_D^{(k)})$, from the composition inequality there are at least $\Lambda_{\delta_1^{(k)},\delta_2^{(k)},\cdots,\delta_D^{(k)}}$ codewords with composition $(\delta_1^{(k)},\delta_2^{(k)},\cdots,\delta_D^{(k)})$ such that all of them do not have a prefix in the previous set of codewords. Hence, the constructed code is a prefix code with composition set $\Delta$.        
\end{proof}

\begin{example}
Let $\Delta:=\{(2,0),(1,1),(3,1)\}$ (where $(a,b)$ represents the composition of a codeword  with $a$ zeros and $b$ ones) from Theorem \ref{1} the existence of a binary prefix code with composition set $\Delta$ is guaranteed because, 

\begin{eqnarray}
\binom{2+0}{2}=1 &\geq& \Lambda_{2,0}\binom{0}{0}=1 \nonumber \\
\binom{1+1}{1}=2 &\geq& \Lambda_{1,1}\binom{0}{0}=1 \nonumber \\
\binom{3+1}{3}=4 &\geq& \Lambda_{3,1}\binom{0}{0}+\underbrace{\Lambda_{1,1}}_{=1}\binom{2+0}{2}+\Lambda_{2,0}\binom{1+1}{1}=4. \nonumber 
\end{eqnarray}   
For example, $\{00,01,1000\}$ is a binary prefix code with composition set $\Delta$. Now, suppose that one more composition $(1,1)$ is also added to $\Delta$, so define $\Delta':=\{(2,0),(1,1),(1,1),(3,1)\}$ then, there is not any binary prefix code with composition set $\Delta'$  because 

\begin{eqnarray}
\binom{3+1}{3}=4 &\ngeq& \Lambda_{3,1}'\binom{0}{0}+\underbrace{\Lambda_{1,1}'}_{=2}\binom{2+0}{2}+\Lambda_{2,0}'\binom{1+1}{1}=5. \nonumber
\end{eqnarray}
\end{example}

From now on all the results are presented for binary codes. In the following theorem we present a polynomial algorithm to find a prefix-free code with the same composition set as a  given code $S$, if such a prefix-free code exists. 

\begin{definition}
For any word $s$ and two numbers $a$ and $b$, $f_{s,a,b}$ is equal
to the number of codewords such as $s'$ with $a$ zeros and $b$ ones such
that $s$ is a prefix of $s'$.
\end{definition}

\begin{theorem}
For any code $S=\{s_1, s_2, \dots, s_n\}$, there is a
polynomial time\footnote{In terms of $n$ and the sum of the lengths of the $n$ codewords.} algorithm which finds a prefix-free code with the same
composition set as the given code $S$, if there exists such a  prefix-free code.
\end{theorem}

\begin{proof}
Without loss of generality, suppose $|s_1| \leq |s_2| \leq \dots
\leq |s_n|$, where $|w|$ is the length of $w$. Our
algorithm has $n$ iterations. In the $i$th iteration, we find a
string $s'_i$ such that the composition of $s'_i$ is the same as the composition  of
$s_i$ and $s'_j$ is not a prefix of $s'_i$ for any $j<i$, as
follows. After $n$th iteration we reach the desired code $S'=\{s'_1,
s'_2, \dots, s'_n\}$ with the same  composition set as the code $S$, and furthermore it is a
prefix-free code.

 Let $a$ and $b$ be the number of zeros and ones in $s_i$,
respectively. If $\Sigma_{j=1}^{i}f_{s_j,a,b} > \binom{a+b}{a}$, then there is
not a code such as $S'$ with the desired properties. Otherwise, there is
a string such as $s'_i$ with the mentioned conditions. We can find the
smallest string such as $s'_i$ in polynomial time as
follows. We iteratively find the bits/digits (code character in binary case) of $s'_i$. For any string such as
$x$ we can check whether  there is a string such as $y$ with the same composition set as $s_i$ such that $x$ is a prefix of $y$ and $s_j$ is
not a prefix of $y$ for any $j<i$. Existence of such a string is
equivalent to this property that the sum of $f_{z,a-c,b-d}$ for all codewords
such as  $z$ for which $s_j=xz$, for some $j<i$ (the notation $xz$ is a concatenation of two codewords $x$ and $z$) is less than all  the codewords
such as $w$ with $a-c$ zeros and $b-d$ ones ($c$ and $d$ are the
number of zeros and ones in $x$, respectively). Now, for finding the smallest $s'_i$, we check whether there is a
$s'_i$ which starts with $0$. If there is such a string, we set the
first bit of $s'_i$ zero. Otherwise, we set it one. Suppose we
have set the first $l$ bits of $s'_i$ and we want to set the $l+1$th
bit. We construct the string  $x$ by concatenating these $l$ bits. We
check whether there is a string such as $y$ such that its composition  is  the same
as $s_i$ and $x0$ is a prefix of $y$ and $s_j$ is not a prefix of
$y$ for any $j<i$. If there exists such a string then the $l+1$th bit is
zero. Otherwise, the $l+1$th bit is one. After $|s_i|$ iterations we
find the desired $s'_i$. 

If there exists a code $S'$ which its composition set is the same as the
composition set of the code $S$ and $S'$ is prefix-free, iteratively as explained in the above, we can find it. Note that our algorithm has $n$ iterations, and in each of these iterations we are computing the sum of at most $n$ values of function $f$. All these operations can be done in time polynomial of $n$ and the sum of the lengths of the codewords. 
\end{proof}
\end{section}
\begin{section}{Fix-free codes}
\label{sec2}

In  Theorem \ref{thasli}, we introduce a sufficient condition under which for a class of codes that we call  \textit{distinct} codes, there exists a fix-free code with the same composition set as the composition set of a given code.


\begin{definition}
A code  $S=\{s_1, s_2, \dots, s_n\}$ is
\textit{distinct} if for any $1 \leq i, j \leq n$, $a_i$ and $a_j$,
satisfy one of the following properties ($a_k$ is the length of the codeword $s_k$ for any $k=1,2,\cdots,n$) :
\begin{itemize}
\item
$
a_i = a_j
$
\item
$
 2a_i \leq a_j
$
\item
$
 2a_j \leq a_i
$
\end{itemize}

It means that if any two codewords $s_i$ and $s_j$ do not have the same size, the size of one of them should be at least twice the size of the other one.
\end{definition}

In the following sequence of lemmas, we present some combinatorial facts that we refer to them along the proof of Theorem \ref{thasli}.

\begin{lemma}
\label{prefix} For a string $s$ with $c$ ones and $d$ zeros, the number of strings
which have $a$ ones and $b$ zeros, and $s$ is a prefix of them is
equal to $\binom{a+b-c-d}{a-c}$, i.e $f_{s,a,b} = \binom{a+b-c-d}{a-c}$.
\end{lemma}

\begin{lemma}
\label{suffix} For a string $s$ with $c$ ones and $d$ zeros, the number of strings
which have $a$ ones and $b$ zeros, and $s$ is a suffix of them is
also equal to $\binom{a+b-c-d}{a-c}$.
\end{lemma}

\begin{lemma}
\label{fix} For any two strings $s_1$ with $c$ ones and $d$ zeros
and $s_2$ with $e$ ones and $f$ zeros, the number of strings which
have $a$ ones and $b$ zeros, and $s_1$ is a prefix of them, and also
$s_2$ is a suffix of them, is equal to $\binom{a+b-c-d-e-f}{a-c-e}$ if we
know that $a \geq c+e$ and $b \geq d+f$.
\end{lemma}

\begin{proof}
Let $s'$ be one of these strings. We also know that $a+b \geq c+d+e+f$. The first $c+d$ letters of $s'$ are fixed because $s_1$ is a prefix of $s'$. The last $e+f$ letters of $s'$ are also fixed because $s_2$ is a suffix of $s'$. It remained to count the number of ways we can fix the rest of the letters of $s'$ such that $s'$ has $a$ ones and $b$ zeros. Note that $s'$ already has $c+e$ ones, and $d+f$ zeros. So we have to put $a-(c+e)$ ones, and $b-(d+f)$ zeros in the rest of the letters (the unfixed letters). This can be done in 
$\binom{a-(c+e)+b-(d+f)}{a-(c+e)} = \binom{a+b-c-d-e-f}{a-c-e}$ ways.
\end{proof}

\par
In [6] it is shown that for any distinct code $S=\{s_1,s_2, \cdots, s_n\}$ satisfying the inequality
$\sum_{i=1}^n 2^{-|s_i|} \leq 3/4$, there is a binary fix-free code with the same codeword lengths.
 In the following, we present a polynomial time algorithm which finds a fix-free code with the same set of composition codewords as the given code $S$, if there exists such a code.

\begin{theorem}
\label{thasli}
For any \textit{distinct} code $S$ with $n$ codewords $s_1, s_2, \dots, s_n$,
there is a polynomial time algorithm which finds a fix-free code
with the same set of composition codewords as the given code $S$, if there exists such a code.
\end{theorem}

\begin{proof}
Without loss of generality, suppose $a_1 \leq a_2 \leq \dots \leq
a_n$, where $a_i$ is the length of $s_i$, $i=1,2,\cdots,n$. Our algorithm has $n$ iterations. In the $i$th iteration, we
find a string $s'_i$ such that composition set of $s'_i$ is same to
composition set of $s_i$ and $s'_j$ is neither a prefix of $s'_i$ nor a
suffix of it for any $j<i$, as follows. After $n$th iteration we
reach the desired code $S'=\{s'_1, s'_2, \dots, s'_n\}$ such that its
composition set is as same as code $S$ and is fix-free. Let $a$ and $b$ be
the number of zeros and ones in $s_i$ respectively. Now, we want to
count the number of strings with $a$ ones and $b$ zeros which are
neither a prefix nor a suffix of any of the strings $s'_1, s'_2,
\cdots, s'_{i-1}$. Note that we can calculate this number only with
knowing the fact that the composition set of each $s'_j$ is exactly the
one of $s_j$, $j<i$. This means that this number  depends only on the
number of  ones and zeros of the previous strings. Now we derive
the number as follows. The number of strings with $a$ ones and $b$ zeros is equal to
$\binom{a+b}{a}$. We decrease the number of strings which have $a$ ones
and $b$ zeros, and $s'_j$ is a prefix of them. We do this decreasing
process for any $j<i$. We also decrease the number of strings which
have $a$ ones and $b$ zeros, and $s'_j$ is a suffix of them. Again
we do this decreasing process  for any $j<i$. According to the fact
that we know the numbers of ones and zeros of $s'_j$ and using
Lemmas \ref{prefix} and \ref{suffix}, we can calculate these
numbers. Now, note that some strings might be decreased twice. For example
for a string $s$ we might have that $s'_j$ is its prefix and also
$s'_k$ is its suffix for some $j,k<i$. But there is no string such as
$s$ that two strings such as  $s'_j$ and $s'_k$ are its prefix at the
same time, because it means that one of these two strings is a prefix
of another which contradicts the fact that none of the strings $s'_1,
s'_2, \cdots, s'_{i-1}$ is a prefix or suffix of another. We can
also conclude that there is no string such as $s$ that two strings such as
$s'_j$ and $s'_k$ are its suffix at the same time. Therefore we just
need to add the number of strings with $a$ ones and $b$ zeros that
$s'_j$ is its prefix, and $s'_k$ is also its suffix for any pair of $j,k$ where $1 \leq j,k < i$. Now for calculating the number of strings which have $a$ ones and
$b$ zeros, and $s'_j$ is their prefix, and $s'_k$ is their suffix,
we have two cases. At first, we suppose that one of these two strings, $s'_j$ and
$s'_k$, has the same length of $s_i$. Without loss of generality
suppose $a_j=a_i$. Now we assert that there is no string such as $s$
that $s'_j$ is its prefix and $s'_k$ is its suffix. Otherwise,
according to the fact that the length of $s'_j$ is equal to $a+b$
which is the length of $s_i$ and $s$, we conclude that $s$ is equal
to $s'_j$. We also know that $s'_k$ is a suffix of $s$ and also is a
suffix of $s'_j$ which contradicts the fact that none of the strings
$s'_1, s'_2, \cdots, s'_{i-1}$ is a prefix or suffix of another.
Therefore there is no such string and our desired number is zero. The other case occurs when the length of both $s'_j$ and $s'_k$ are
strictly less than the length of $s_i$. Using the fact that our code
is \textit{distinct} we conclude that $2|s'_j| \leq a+b$ and $2|s'_k| \leq
a+b$, so we have $|s'_j|+|s'_k| \leq a+b$. Now we can apply Lemma
\ref{fix}, and calculate our desired number. According to the Inclusion and Exclusion principle we should
continue this process of decreasing and increasing iteratively, but
actually we do not need to do it anymore, because there is not any
string such as $s$ such that three strings like $s'_j$, $s'_k$ and
$s'_l$ are either its prefix or its suffix. The reason is somehow
clear, because if there were the three strings $s'_j$, $s'_k$ and $s'_l$
which are either a prefix or a suffix of $s$, then according to the
pigeonhole principle two of them should be a prefix of $s$, or two
of them should be a suffix of $s$. In the former case, we see that
one of the strings $s'_j$, $s'_k$ and $s'_l$ is a prefix of another,
and in the latter case, we see that one of the strings $s'_j$,
$s'_k$ and $s'_l$ is a suffix of another. But this again contradicts
the fact that none of the strings $s'_1, s'_2, \cdots, s'_{i-1}$ is a
prefix or suffix of another. So, using this algorithm, we can iteratively count the number of
choices we have to replace with $s_i$. If this number is zero in one
step, this means that there does not exist such a fix-free code. But,
if this number is greater than zero in each iteration, we have some
choices in each iteration and, finally we reach a fix-free code.

So, for string  $s_i$ we count the number of strings like $s'_i$ with the same composition set of $s_i$ such that no $s'_j$ (for $j<i$) is neither a prefix of $s'_i$ nor a suffix of $s'_i$. We can compute this number as follows: 

\begin{eqnarray}
\binom{a+b}{a} &-& \sum_{1 \leq j < i} \mbox{PrefixNum}(s_i,s'_j) -\sum_{1 \leq j < i} \mbox{SuffixNum}(s_i,s'_j) \nonumber \\
               &+& \sum_{1 \leq j,k < i} \mbox{PrefixSuffixNum}(s_i,s'_j,s'_k)  
\end{eqnarray}
In  above formula, $\mbox{PrefixNum}(s_i,s'_j)$ is the number of strings like $s'_i$ with the same composition set of $s_i$ such that $s'_j$ is its prefix. Similarly, $\mbox{SuffixNum}$ is defined. We also define $\mbox{PrefixSuffixNum}(s_i,s'_j,s'_k)$ to be the number of strings like $s'_i$ with the same composition set of $s_i$ such that $s'_j$ is its prefix, and $s'_k$ is its suffix. Note that the above formula is basically the simplified version of  Inclusion  Exclusion Principle knowing the fact that there can not be three strings among $s'_1, s'_2, \cdots, s'_{i-1}$ such that each of them is either a prefix or a suffix of the same string. 

If this number is positive we know that there exists a string $s'_i$ with the desired properties. But we have to find this string as well. This is done by searching in the binary search in the tree of all strings. 
Here we show that we can find the lowest string (alphabetically) $s'_i$ with these properties. 
At first we try to find a string $s'_i$ that starts with zero.  We count all strings $s'_i$ with the desired properties that also start with zero. This can be done by changing each term in the above formula by assuming that $s'_i$ starts with zero. For example, instead of $\binom{a+b}{a}$ we should write $\binom{a+b-1}{a}$. If $s'_j$ starts with one, the number $\mbox{PrefixNum}(s_i,s'_j)$ should be replaced with zero because we know that $s'_i$ is supposed to start with zero, and therefore $s'_j$ can not be its prefix. So, we change the above formula, accordingly. If the number of these strings is positive, we know that there exists an string $s'_i$ with the desired properties that also starts with zero. So, we fix the first digit to be zero, and go on to the next digit. We can iteratively continue this process till there are $a$ ones and $b$ zeros in our string. This can be done by computing the above formula $a+b$ times (in each iteration we fix a digit). 

Our algorithm runs in polynomial time in terms of $n$ and the total number of ones and zeros in all $n$ input strings.   
\end{proof}

In  Lemma \ref{4app}, a polynomial time algorithm is provided to find a near optimal fix-free code when its maximum cost and the number of codewords are given. To the best of our knowledge, it is the first approximation algorithm for this problem. We assumed (without loss of generality) that the cost of a zero is $1$ and the cost of a one is $m \geq 1$. \\

Notice that in the case when the letter costs are equal, i.e. $m=1,$ it is known that (\cite{ahl}) for each probability distribution $P=(p_1,p_2,\cdots,p_n)$ there exists a fix-free code where the average cost of the codewords is bounded above by $H(P)+2$, where $H(P)=-\sum_{i=1}^np_i \log p_i$ is the entropy of the source.
In the following  lemma the objective is to minimize the average codeword cost (defined in \eqref{eqn1}) for \textit{general letter costs} with \textit{uniform distribution} of the codewords.

\begin{lemma}
\label{4app} 
For any given number $x$, if there exists a fix-free code
such as $S$ with $n$ codewords and cost at most $x$, we can find a fix-free
code in polynomial time with cost at most $(5+{1 \over n-1})x$. 
\end{lemma}

\begin{proof}
Let $y$ be $x/n$. Note that $y$ is the mean cost of the $n$ codewords
in $S$. So the number of codewords with cost more than $2y$ is less than
$n/2$ and the number of codewords with cost at most $2y$ is at least
$n/2$. Because if there are more than $n/2$ codewords in $S$ with cost at least $2y$, the total cost of $S$ would be more than $n/2 \times 2y = n \times y = x$ which is a contradiction. Let $A$ be the number 
of codewords in $S$ with cost at most $2y$. We conclude that 
$A$ is at least $n/2$. Name these $A$ codewords
$s_1, s_2, \cdots, s_A$.   

 These codewords have at most $l=\lfloor 2y\rfloor$
letters(including zeros and ones) and at most $k=\lfloor 2y/m
\rfloor$ ones (because zero has cost $1$, and one has cost $m$). 
Let $A$ be the number of codewords with at most $l$
letters and $k$ ones). 

Now we change these $A$ codewords in the following way to get $A$ new codewords that have the same size, and are also fix-free. 

If some of these codewords have less than $l$
letters, we add some zeros to their ends in order to make all of
them have the same length, $l$. 
So we add $l-|s_i|$ zeros at the end of $s_i$ where $|s_i|$ is the length of $s_i$. Let $s'_i$ be the new codeword. Clearly we get $A$ codewords $s'_1, s'_2, \cdots, s'_A$ with the same size, $l$. We now prove that these $A$ new codewords are different by contradiction. 

Assume two codewords $s'_i$ and $s'_j$ are the same. Without loss of generality, assume that $|s_i| \geq s_j$. Since $s'_i$ is the same as $s'_j$, the codeword $s_j$ is a prefix of $s_i$ which is a contradiction. Because codewords $s_1, s_2, \cdots, s_n$ come from a fix-free code, so none of them can be a prefix of another. 
So the codewords $s'_1, s'_2, \cdots, s'_A$ are not equal to each other at all.

Now we can get $2A$ codewords which form a fix-free code with some modifications as follows. 
For each codeword $s'_i$, add a zero at the end of $s'_i$, and get the new codeword $s'_{0,i}$. In the same way add a one at the end of $s'_i$, and get the new codeword $s'_{i,1}$. Now we have $2A$ codewords $s'_{1,0}, s'_{2,0}, \cdots, s'_{A,0}, s'_{1,1}, s'_{2,1}, \cdots, s'_{A,1}$ each of which has size $l+1$. Since the $A$ codewords $s'_1, s'_2, \cdots, s'_A$ are $A$ different codewords, these $2A$ codewords are also different, and have the same size, so none of them is a prefix or suffix of another one. 

Since $2A$ is at least $n$, 
we conclude that there exists $n$ codewords with length $l+1$ and at most $k+1$ ones in each of the codewords. 

Let $T$ be the set of all codewords with length $l+1$ and at most $k+1$ ones. We proved that there are at least $n$ codewords in $T$.
We just need to pick $n$ arbitrary codewords from $T$ (one can start from the codewords with one $1$, and then two $1$s, and so on, and pick $n$ codewords this way). Since all members of $T$ have the same size and two different codewords with the same size can not be prefix or suffix of each other,  the result of our algorithm would be fix free. 

Now we analyze the cost of the code we obtained. 
The cost of these $n$ arbitrary codewords is at most 
$[(k+1)m+(l-k)]n$. The ratio of this
cost to the optimal cost $x$ is ${[(k+1)m+(l-k)]n \over x}={kmn
\over x } + {ln \over x} + {(m-k)n \over x} \leq 2 + 2 + {mn \over x}
\leq 4 + 1+ {1 \over (n-1)} = 5+\frac{1}{n-1}$. 
Note that we defined $l$ and $k$ such that $ln \leq 2x$, and $kmn \leq 2x$. 
We also know that there are at most one word in the
optimal fix-free code that does not have any one. 
So there are $n-1$ codewords in optimal code that each of them has at least one $1$. So the cost of optimum, (which is at most $x$), is  at least $(n-1)m$ and therefore ${mn \over x}  \leq 1 + {1 \over n-1}$. So we proved that the cost of our code is  at
most $[5 + 1/(n-1)]x$.
\end{proof}

Note that when there does not exist a fix-free code with cost at most $x$, the algorithm in
 Lemma \ref{4app} may return a code with cost at most $(5+{1 \over n-1})x$ or fail.
 
 Furthermore, it is useful to add that Lemma \ref{4app} fails if and only if the set $T$ contains less than $n$ codewords, and that, as $x$ increases, the size of $T$  does not decrease, therefore, if the algorithm is successful for some $x$, then it will be successful for all values larger than $x$.\\
 
 In the following theorem, we present  an approximation algorithm that always finds a
fix-free code such that its cost is at most $5+{1 \over n-1}+\epsilon$ times the cost of the optimal code.

\begin{theorem}
For any $n$ and $\epsilon>0$, there is a $5 + {1 \over
n-1}+\epsilon$-approximation algorithm for the problem of finding
the optimal fix-free code with $n$ codewords such that its time
complexity is a polynomial of the $n$ and $1 \over \epsilon$.
\end{theorem}

\begin{proof}
Let $y$ be the cost of the optimal fix-free code. If we know the value of $y$, we can find a fix-free code with cost
at most $(5+{1 \over n-1})y$ using Lemma \ref{4app}, and the claim is true.
Although $y$ is not given as an input,
we can guess the $y$ by a typical binary
search and with error $\epsilon$ by guessing $O(\log(n(n+m)/\epsilon))$ times.
Actually we know that $y$ is at least $n$.
We also know that $y$ is at most $n(n-1+m)$ because
there are exactly $n$ codewords which have only one $1$ and $n-1$ zeros.
These codewords form a fix-free code and the cost of this code is
$n(n-1+m)$. So we have $n \leq y \leq n(n-1+m)$.
Let $x$ be the minimum number for which the algorithm in
Lemma \ref{4app} returns a code with cost at most $(5+{1 \over n-1})x$.
We are going to find $x$ with error $\epsilon$.
We know that $x \leq y$ and $0 \leq x \leq n(n-1+m)$.
we are going to run a binary search in the interval
$[0,n(n+m-1)]$. In each step, we can decrease the length of our interval to half of its previous length.
For example, if we know that $x$ is in $[\alpha, \beta]$, we define $z$ to be ${\alpha+\beta \over 2}$. Next using
Lemma \ref{4app}, we can know that whether $x \leq z$ or not, because if the algorithm in Lemma \ref{4app} fails, $x$ is
greater than $z$. Otherwise, $x$ is at most $z$. So after each step we know that $x$ is in
$[\alpha , {\alpha+\beta \over 2}]$ or $[{\alpha+\beta \over 2} , \beta]$. Therefore the length of our searching interval
is multiplied by $1 \over 2$ in each step, and after $\log(n(n+m)/\epsilon)$ steps the length of our interval is
at most $\epsilon$.
Because at first the length is less than $n(n+m)$. Finally we know that $x$ is in $[t,t+\epsilon]$ where the algorithm
in Lemma \ref{4app} does not fail for $t+\epsilon$. In the other words we can find a fix-free code with cost at
most $(5+{1 \over n-1})[t+\epsilon]$. As we know $t+\epsilon \leq x+\epsilon \leq y+\epsilon$. We conclude that the
fix-free code that we just found has a cost of at most $(5+{1 \over n-1})[t+\epsilon] \leq (5+{1 \over n-1})[y+\epsilon]
(5+{1 \over n-1} + \epsilon)y$ because $y$ is at least $n$. Therefore we found a fix-free code with cost at most $(5+{1 \over n-1}+\epsilon)$
times the cost of the optimal code.

\end{proof}

\end{section}



\subsection*{Acknowledgments}
The authors are indebted to  Serap Savari for stimulating and productive discussions
for this work.


\end{document}